\documentclass[letterpaper,10pt,twocolumn,conference]{ieeeconf}  

\IEEEoverridecommandlockouts                              
\usepackage{color}

\usepackage{graphics} 
\usepackage{epsfig} 
\usepackage{amsmath} 

\usepackage{amsthm}

\usepackage{floatrow}
\usepackage{placeins}
\newtheorem{thm}{Theorem}
\newtheorem{lem}{Lemma}
\usepackage{caption}
\usepackage{subcaption}
\usepackage{algorithm}
\usepackage[noend]{algpseudocode}
\usepackage{authblk}
\usepackage{hyperref}
\usepackage{setspace}
\DeclareMathOperator*{\argmax}{argmax}

\begin{document}
\title{\LARGE \bf
A Bounded Multi-Vacation Queue Model for Multi-stage Sleep Control 5G Base station
}
\author{Jie Chen\thanks{Jie Chen is a postdoctoral associate with the physics department of Durham University, Newcastle and Durham Joint Quantum Centre (e-mail:jie.chen@durham.ac.uk)}}
\maketitle

\thispagestyle{empty}
\pagestyle{empty}

\begin{abstract}
Modelling and control of energy consumption is an important problem in telecommunication systems.To model such systems, this paper publishes a bounded multi-vacation queue model. The energy consumption predicted by the model shows an average error rate of 0.0177 and the delay predicted by the model shows an average error rate of 0.0655 over 99 test instances.Subsequently, an optimisation algorithm is proposed to minimise the energy consumption while not violate the delay bound. Furthermore, given current state of art 5G base station system configuration, numerical results shows that with the increase of traffic load, energy saving rate becomes less.
\end{abstract}

\section{INTRODUCTION}
Bounded multi-stage sleep mode control has emerged as a future implementation feature for energy efficient 5G networks.In this scheme, mobile devices hibernate gradually from light to deep sleep through a limited number of discrete stages and resume to work when there arrives a new workload or when they finish hibernation.To analyse this scheme, queuing theory is being referred to in this paper. Vacation queue system has been in discussion in literature for long. These systems are working upon the policy whether the number of packets in queue reaches a threshold or not ($N > 0$) or whether the vacation time has exceeded certain amount (T policy) ~\cite{doshisurvey}. In the following subsections, the feasibility of this new scheme and the supremacy of it against other options are being discussed. 

\subsection{The Merits of BMV-policy over other policies}
Investigations have demonstrated the convincing results that Bounded Multi-vacation Queue (BMV) schemes can beat N-policy and T-policy schemes in terms of system performance and reliability. As N-policy has only adjustable parameters of $K$ - the system buffer maximum quota, it has a bounded energy consumption rate and delay while BMV-policy can tune its $N_v$ (vacation amount limit) and $L_v$ (vacation length mean) across much wider ranges to guarantee an improved solution. Similarly T-policy can be treated as a single vacation(SV) policy, given a fixed $T$, results have shown that if being broken into multiple equally weighted vacations to make a BMV-policy scheme, the system would achieve a much smaller delay.
Figure \ref{fig:n-policy} coincides with paper ~\cite{npolicyJie} that power consumption level fluctuates and delay increases with the increase of N. In this particular simulation, $\lambda=550$, $\mu=1000$,$K=50$,$power_{on}=130$,$power_{off}=75$. Results show that though $\rho=0.55$ and $ratio_{power}=0.5769$, the normalised energy consumption per bit for N-policy regardless of which $N$ is selected, goes much higher than $0.6$. The system has not been saturated in terms of energy conservation efficiency. As delay is a traditional QoS metric of a network system, we might also want the new scheme would outperform N-policy scheme in terms of processing speed. Given a bounded delay for N-policy as $[D_{min},D_{max}]$,solutions with BMV-policy can be easily founded that match the design criteria that consumes less energy while falls within the delay bounds. Two of them are depicted on the figure as examples.
\begin{figure}[!htbp]
    \centering
    \includegraphics[width=0.98\textwidth]{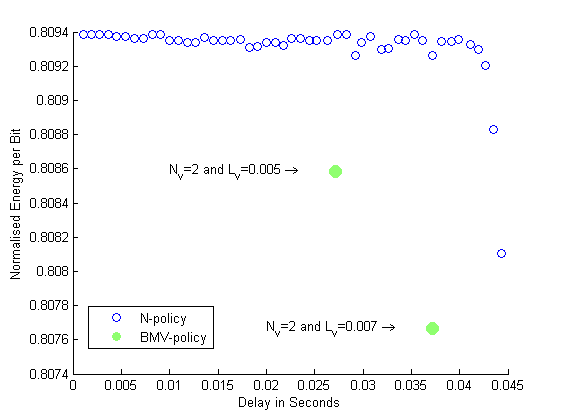}
    \caption{BMV-policy vs N-policy (Simulation)}
    \label{fig:n-policy}
\end{figure}
\FloatBarrier

Suppose in T-policy, the vacation length is $L_v$ and in BMV-policy, with the increase of $n$ (maximum number of vacations), $L^{BMV}_v=\frac{L_v}{n}$, cases where $n \in [1\ 7]$ are being executed and evaluated.Figure \ref{fig:t-policy}, shows that with the increase of $n$, the delay decreases while the energy level fluctuates.Based on the limited tested cases, $NE_{BMV}>NE_{T}$ with $NE$ stands for normalised energy while $D_{BMV}<NE_{T}$. Given a fixed delay bounds $[D_{min} \ D_{max}]$ imposed by N-policy, results show that BMV-policy can produce feasible solutions with higher energy savings.                                 

\begin{figure}[!htbp]
    \centering
    \includegraphics[width=0.98\textwidth]{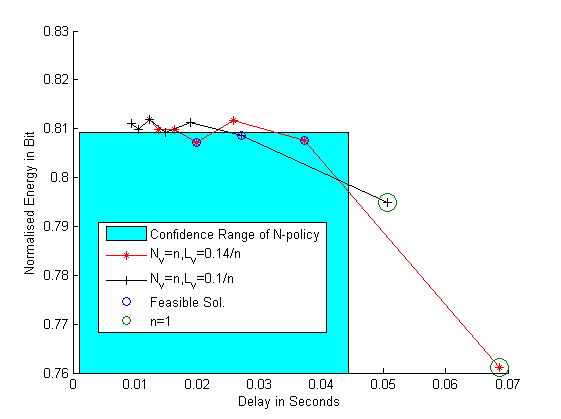}
    \caption{BMV-policy vs T-policy (Simulation)}
    \label{fig:t-policy}
\end{figure}
\subsection{Contribution}
\begin{itemize}
    \item The following work is the first to theoretically discuss multi-stage sleep control in the current state of art 5G Base station.
    \item It bridges the mathematically theoretical analysis and the practical engineering problem by the validation of software simulation. Previous works in theoretical queuing analysis rarely endeavor to go through thorough experimental tests. Neither have works in mobile engineering set forth to develop an analysis from brand new.
    \item It treats the system design problem as an optimal control problem considering the trade-off between delay and energy consumption and provides sound analysis against both of these system metrics. Most of the new queuing analysis are devised to evaluate the delay metrics solely and for those to have evaluated the cost metric such as power consumption in a typical telecommunication system,this work is the first to propose a validated model for future accurate prediction.  
    \item It pioneers in evaluating the performance of the multi-stage sleep control 5G system using the validated model. Moreover, it is the first to evaluate tentatively the selective service rate scheme given a fixed pool within such system.
\end{itemize}
\section{Problem Formulation}
In paper \cite{npolicyJie}, the system is perceived to rotate between sleep mode and working mode. Following this approach, the system performance is evaluated such that instead of me focusing on an equilibrium long term with system running a countably infinite time frame by simulation, the system running thread is composed of multiple running cycles. In this paper, by averaging over these running cycles, a particular uniform cycle is inspected that consist of a sleeping sub-frame and an active sub-frame, the statistically distributed measurements such as power consumption and delay are calculated and equated to those in longer term. 
\subsection{System Description}
The queuing system consists of an intelligent server that can vacate whenever the queue is empty. The vacation duration is adjusted based on two parameter configurations. They are the maximum vacation number $N_v$ and the average vacation period $L_v$. To be more specific, the queue once in vacation mode will return to the workstation whenever a vacation period expires. If the queue is still empty, it will continue to next vacation period till the maximum vacation number is reached. Otherwise it will resume to work upon its return to the workstation. Please refer to Figure ~\ref{fig:flow} for further illustration;
\begin{figure}[h]
    \centering
    \includegraphics[width=0.98\textwidth]{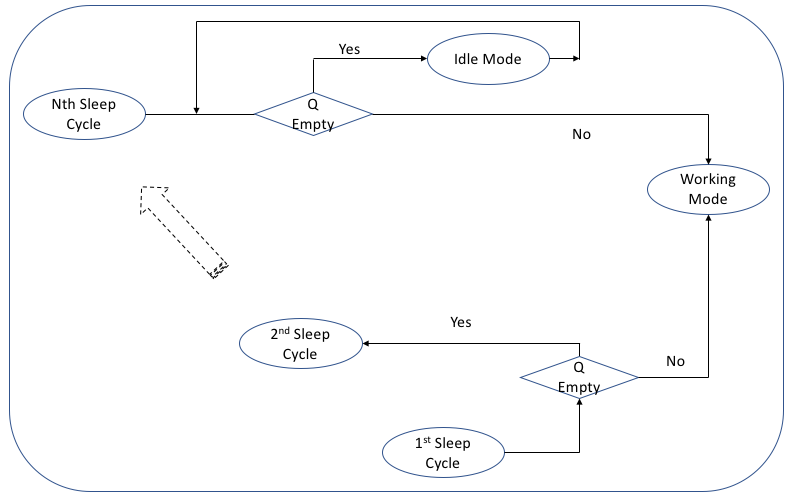}
    \caption{System Work Flow} 
    \label{fig:flow} 
\end{figure}
\FloatBarrier
The input traffic model follows a Poisson distribution with an average rate of $\lambda$ and the service pattern follows an Exponential distribution with an average rate of $\mu$. For the time being, the power that the sever uses at vacation is $p_s$, the vacation length is uniform over all stages and the power that the server uses at work is $p_a$.

\subsection{Power Consumption Analysis}
 Assume that the sleeping sub-frame has a length of $L_s$ and the working sub-frame has a length of $L_b$.
By design it is also assumed that the system starts with sleeping mode.

As there are at most $N_v$ sleeping periods, for the system to enter working mode after the first period, there must be at least 1 arrival during the first period. For the system to enter working mode after the $n$ period with $n \le N_v$, there will be at least 1 arrival during the previous $(n-1)th$ period but none happens during the previous $(n-2)$ periods.

Let A presents the event that there is at least 1 arrival within time frame $iL_v$; B presents the event that there is at least 1 arrival within time frame $(i-1)L_v$; C presents the event that there is at least 1 arrival within time frame $iL_v-(i-1)L_v$ only, which is equivalent to the phenomena that packets start to arrive during $i^{th}$ sleeping period. Henceforth $C=A-B$. Further on, $P(C)=P(A)-P(B)$. The formula is constructed as :
\begin{equation*}
    P_{L_s}(i)=\exp(-\lambda (i-1) L_v)-\exp(-\lambda (i) L_v)
\end{equation*}
Upon entering the working period, the system has a starting probability distribution of queue length which is $P_{init}=[p_k]\times K$ with $K$ is the maximum queue size.
$p_k=\exp(-\lambda L_v)  \frac{{(\lambda L_v)}^{k}}{k!}$


\begin{thm}
As the system probabilistically evolves from the initial distribution towards an approximately zero position dominated distribution such that the $P_{\mathrm{other}}(n)$ is approximate to 0. The summation of the $P_{\mathrm{zero}}(n)$ across all the stopping point is approximately to 1 as much as possible, $\forall \epsilon$, $\exists N$, when $n>N$, $1-\sum_k^{n} P_{\mathrm{zero}}(k) \le \epsilon$.
\end{thm}
\begin{proof}
Assume the initial probability distribution is $P_0=[P_0(0) \ P_{\mathrm{other}}(0)]$ and $\overline{P}_0=[0 \ P_{\mathrm{other}}(0)]$, $[P_0(1) \ P_{\mathrm{other}}(1)]=\overline{P}_0*P_{\mathrm{tran}}$.\
\begin{equation*}
    \begin{aligned}
    P\mathrm{sum}_1 &=\sum_{i=0}^1 P_0(i)=P_0(0)+P_0(1)\\
                    &=\sum_{i=1}^K P_{\mathrm{other}}^i(0)-\sum_{i=1}^K P_{\mathrm{other}}^i(1)+P_0(0) \\
                    &=1 -\sum_{i=1}^K P_{\mathrm{other}}^i(1) \le 1
    \end{aligned}
\end{equation*}

\begin{equation*}
    \begin{aligned}
    P\mathrm{sum}_2 &=\sum_{i=0}^2 P_0(i)=P_0(0)+P_0(1)+P_0(2)\\
    &=1-\sum_{i=1}^K \overline{P1}_i +\sum_{i=1}^K \overline{P1}_i -\sum_{i=1}^K \overline {P2}_i\\
    &=1-\sum_{i=1}^K \overline {P2}_i
    \end{aligned}
\end{equation*}
 It can be intuitively derived that 
 \begin{equation*}
     \begin{aligned}
     P\mathrm{sum}_n &=1-\sum_{i=1}^K \overline {P(n)}\\
     &=1-P\mathrm{sum}_{\mathrm{other}}(n)
     \end{aligned}
 \end{equation*}
 
\begin{equation}
    \begin{split}
        \sum_{i=0}^K \overline {P_i(n)}&=\sum_{i=0}^K \sum_{j=1}^K \overline {P_j(n-1)}*P_{\mathrm{tran}}(j,i)\\
        &=1-\sum_{i=0}^K \overline {P_0(n-1)}*P_{\mathrm{tran}}(0,i)\\
        &=\sum_{i=1}^K \overline {P_i (n-1)}\\
        &=\sum_{i=0}^K \overline {P_i(n-1)}-\epsilon_{n}
    \end{split}
\end{equation}
where $\epsilon_{n}= \overline {P_0(n-1)}$. Henceforth, the following equation can be justified that $P\mathrm{sum}_{\mathrm{other}}(n)<P\mathrm{sum}_{\mathrm{other}}(n-1) \le 1$.
$P\mathrm{sum}_{\mathrm{other}}(i)$ is thus a monotonically decreasing sequence while $P\mathrm{sum}_n$ is a monotonically increasing sequence within the frame $[0,1]$.
Hence,given an $\epsilon$ as small as possible, there always exists an $N$ that $N=\argmax_i P\mathrm{sum}_{\mathrm{other}}(i)>\epsilon$, for $n>N$, $1-P\mathrm{sum}_n=P\mathrm{sum}_{\mathrm{other}}(n)<\epsilon$.  $\qedsymbol$ 

\end{proof}

Assume the transition matrix is $P_{\mathrm{tran}}$,the formation of $P_{\mathrm{tran}}$ for current $M/M/1/K$ queue system can be extended from paper \cite{npolicyJie}.  $P_{\mathrm{tran}}=[p_{i,j}]*[K \times K]$
\begin{equation}
    \begin{split}
        p_{i,j}=
        \begin{cases}
            &0  \\
            &\text{if $j<i-1$} \\
            &\int_0^{\infty} \mu exp(-(\lambda+\mu)t)\frac{(t\lambda)^{(j-i+1)}}{(j-i+1)!}\\
            &\text{if $(i-1)<=j<K$} \\
            &\sum_{j=K}^{\infty} \int_0^{\infty} \mu exp(-(\lambda+\mu)t)\frac{(t\lambda)^{(j-i+1)}}{(j-i+1)!}\\
            &\text{ if $j=K$}
        \end{cases}
    \end{split}
\end{equation}
$[P_{\mathrm{zero}}(k) \  P_{\mathrm{other}}(k)]=[0 \  P_{\mathrm{other}}(k-1)]*P_{\mathrm{tran}}$

\begin{equation} \label{equation1}
\begin{split}
E[L_b] & = E[E[L_b|{l_k=\sum_i^k x_i}]  \\
     & = \sum_k^{n>N} P_{\mathrm{zero}}(k) E[l_k]]\\
     &=\sum_k^{n>N} P_{\mathrm{zero}}(k)kE[x_i]\\
     &=\sum_k^{n>N} P_{\mathrm{zero}}(k)k\frac{1}{\mu}
\end{split}
\end{equation}
The special event that no arrival within the maximum number of vacation periods is analysed as below:
The period between the end of the overall sleeping sub-frame and the beginning of server running period is labeled as $ilen$ - the idle length.
As Poisson Distribution follows an individually independent Markovian pattern, $ilen$ is perceived as the inter-arrival time between the zero arrival and the first arrival minus the maximum overall sleeping sub-frame.
\begin{equation}
\begin{split}
E[L_i]&=E[L_a|N_{L_s}=0]\\
&-L_vN_v\\
&=\int_{L_vN_v}^{\infty} \frac{t\lambda exp(-t\lambda)}{exp(-(L_vN_v)\lambda)} \\
&- L_vN_v
\end{split}
\end{equation}

$L_a$ is the inter-arrival time for the first packet in the idle mode and $N_{L_s}$ is the number of arrivals within period $L_s$.

At the end of this inter-arrival time, the queue length probability distribution as $P_{\mathrm{init}}$ is $[0]\times K$ and $P_{\mathrm{init}}[1]=1$

Let the ratio $r=\frac{L_s}{L_b+L_s}$. Normalised energy per bit can be derived from $E_i=1-r+r*\frac{p_s}{p_a}$. $i<=N_v$ are the events where the server resume to work within the maximum amount of sleep frames. In these cases $L_s=i*L_v$ and $L_b+L_s=L_b+i*L_v$. Then $i=N_v+1$ is the event where the server has an idle stage between the sleeping sub-frame and the working sub-frame. In this case $L_s=N_v*L_v$ and $L_s+L_b=L_b'+ilen+N_v*L_v$. Lastly the event for $i>N_v+1$ doesn't exist. Let $NE$ be the acronym for normalised energy per bit, then.
\begin{equation}
\begin{split}
E[NE]&=\sum_i^{N_v} E_i*P_{L_s(i)} +\\ &E_{N_{(v+1)}}exp(-\lambda N_vL_v)
\end{split}
\end{equation}
where the probability of first arrival within $i^{th}$ vacation period $P_{L_s(i)}=\exp{-\lambda (i-1)L_v}-\exp{-\lambda i L_v}$.
\subsubsection*{Case Study}
$\mu=0.8$,$N_v=4$,$L_v=\{a|a=\frac{1}{L_v}=0.1+0.05*i,\ i\in [1,9] \text{and}\  i \in \mathrm{Z}\}$ 
From Figure ~\ref{fig:power}, it can be noticed that the analytical plots based from the above procedure have the similar curve as the simulation plots and the numerical values are pretty close to each other as the average error rate is $0.0177$ and deviation is $0.0102$ over 99 data instances.
\begin{figure}[!htbp]
    \centering
    \includegraphics[width=0.98\textwidth]{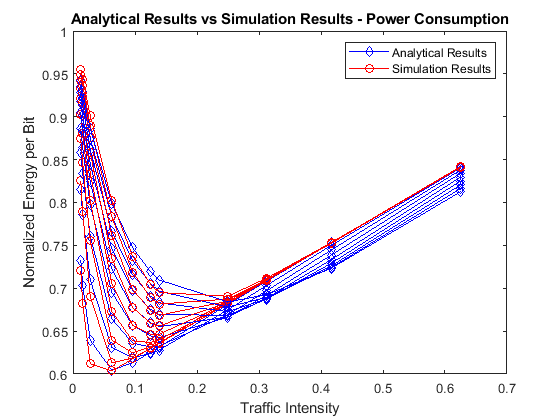}
    \caption{Normalised Power Analysis Validation}
    \label{fig:power}
\end{figure}

\FloatBarrier

\subsection{Waiting Time Analysis}
Waiting time analysis borrows Little's Theorem basic idea.The analysis is performed as decomposing the long term waiting time average for the system into two event cases: $A$ no arrival within the limited vacation time; $B$ no less than 1 arrival within the limited vacation time. It is easy to conclude that $P(A)=exp(-\lambda(N_vL_v))$ and $P(B)=1-P(A)$.

\begin{thm} \label{thm_2} For event $A$, the system is working as a $M/M/1/K$ system without any policy.\footnote{the theory has been similarly mentioned in literature already ~\cite{doshidecompo}. Here a more intuitive and alternative approach is presented.}
\end{thm}

\begin{lem} The waiting time for vacation queuing system in general is equivalent to the waiting time for packets in an averaged running cycle.

\begin{proof}
By Little's theorem, in the long term, the overall packet-in-queue time summation to the number of overall in queue packets is the waiting time. By formula, $W=\lim_{t \to \infty} \frac{\gamma (t)}{\alpha (t)}$, where $\gamma(t)$ is the packet time summation up to time instance $t$ and $\alpha(t)$ is the in-queue packet number summation up to time instance $t$. The overall system time is consisting of infinite number of running cycles. Suppose for an averaged running cycle, the overall packet-in queue time summation is $\triangle \gamma_k$ and the overall in queue packet number is $\triangle \alpha_k$.
\begin{equation*}
    \begin{aligned}
   W&=\lim_{n \to \infty}\frac{\sum_{i=1}^{n} \triangle \gamma_i}{\sum_{i=1}^{n}\triangle \alpha_i}\\
    &=\lim_{n \to \infty} \frac{n\triangle \gamma_i}{n\triangle \alpha_i}\\
    &=\frac{\triangle \gamma_k}{\triangle \alpha_k}
    \end{aligned}
\end{equation*}
$\qedsymbol$ 
\end{proof}
\end{lem}
\begin{lem}
In a no policy $M/M/1/K$ system, the timer of arrival process and departure process are synchronised. The Markov transition diagram can be drawn time-invariably and subsequently the classical equilibrium probability formula can be derived. The event of no arrival within vacation time falls into the category because the arrival process and departure process are synchronised.In this case, the Markov transition diagram starts when the server finishes vacation and embarks on idle period. Henceforth, the waiting time is $W=(\rho*(1+K*\rho ^{K+1}-(K+1)*\rho^K)/((1-\rho)*(1-\rho^{K+1})))*{\lambda}^{-1}$ with $\rho=\frac{\lambda}{\mu}$ and $K$ is the queue limit.
\end{lem}

For event $B$, following \textbf{Theorem \ref{thm_2}}, $W=\frac{\triangle \gamma_B}{\triangle \alpha_B}$. Here $\alpha_B$ the overall packet in queue number is equivalent to the number of packets that have been departed during an averaged running cycle as the running cycle only stops when all the packets in queue are out of the system. $\triangle \gamma_B = A_s+A_b$ where $A_s$ is the packet in-queue time summation during an averaged vacation cycle and $A_b$ is the packet in-queue time summation during an averaged busy cycle.

The conditional queue length when the system resumes to work is $L^{init}_Q=E(L_Q|N_a>0)=\frac{\sum_{i=1}^K iPinit(i)}{\sum_{i=1}^K Pinit(i)}$, where $N_a$ is the number of arrival.

\begin{algorithm}
\caption{Calculation of $A_s$}
\begin{algorithmic}[1]

\Procedure{CalcAs}{$\lambda,L^{init}_Q$}     
    \State $A_s=0$ , $i=0$ \ and $res=L^{init}_Q$
    \If{$L^{init}_Q<1$}
        \State $A_s=\frac{1}{lambda}*res$
    \Else{
        \While{$i \le L^{init}_Q$}
            \If{$res<1$}
                \State $A_s=A_s+ i*\frac{1}{\lambda}*res$
            \Else
                \State $A_s=A_s+i*\frac{1}{\lambda}$
            \EndIf
            \State $i=i+1$
            \State $res=L^{init}_Q-i$
        \EndWhile
    }
    \EndIf
\EndProcedure
\end{algorithmic}
\end{algorithm}

The conditional queue length when the system is active at departure epoch $k$ is $P_{ql_k}=\frac{\sum_i P^k_{\mathrm{other}}(i)i}{\sum_i P^k_{\mathrm{other}}(i)}$ when $k=0$, $Pqlen=L^{init}_Q$, 
\begin{equation}
\begin{aligned}
A_b= \sum_{i=0}^{K \to \infty} P_0(i)A^i_b
\end{aligned}
\end{equation}
$A^i_b$ is the packet in-queue time summation when the queue becomes empty at the $ith$ departure epoch for an averaged busy cycle.
\begin{equation}
    \begin{aligned}
    A^i_b=\sum_{k=0}^i 0.5*((P_{ql_k}+P_{ql_k}+\frac{\lambda}{\mu})/\mu)
    \end{aligned}
\end{equation}
\subsubsection*{Case Study}
The parameters are set in accordance with Section B - Case Study. The analytical plots -Figure ~\ref{fig:delay} have some discernible discrepancies from the simulation results, esp. for $L_v=16.66667$ when high chance of multiple arrivals within the first single sleep vacation exists. It is not an ideally targeted situation for this bounded multi-vacation policy.
The average error rate $\frac{|VAL_{ana}-VAL_{sim}|}{VAL_{sim}}$ over all the 99 instances is 0.0655 and standard deviation is 0.0483. 
\begin{figure}[!htbp]
\centering
    \includegraphics[width=0.98\linewidth]{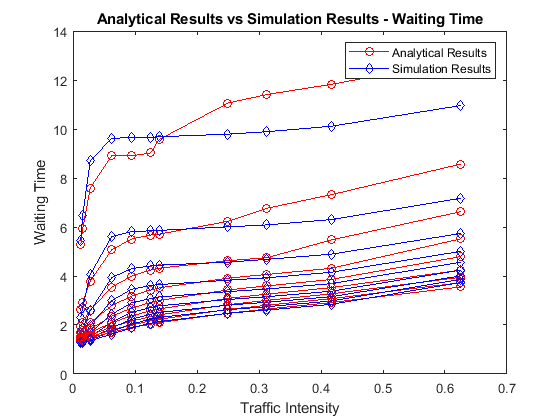}
\caption{Waiting Time Analysis Validation}
\label{fig:delay}
\end{figure}

\subsection{Optimisation Scheme}
The optimisation goal is to select an ideal $(L_v,N_v)$ pair from a feasible pool for a given input traffic rate $\lambda$, given a fixed service rate $\mu$.

\begin{algorithm}
\caption{Search for Optimal Vacation Period and Vacation Maximum Number}
\begin{algorithmic}[1]

\Procedure{OptSearch}{$\lambda,\mu,Dconst$}   \Comment{$Dconst$ is the waiting time bound}    
    \State ${Pool}_{L_v}$,${Pool}_{vnum}$ Initialisation
    \State $minP=1$,$optL_v=0$ and $optVnum=0$
    \While{$Pool_{vnum}$ not exhausted}
        \State $vnumIndex =vnumIndex+1$
      \While{$Pool_{L_v}$ not exhausted}
        \State $L_vIndex = L_vIndex+1$  
        \State $E=Power\_Analysis\_Function$
        \State $W=Waiting_Time\_Analysis\_Function$
        \If{$W<Dconst$}
            \If{$E<minP$}
                \State $minP=E$
                \State optVnum$=Pool_{vnum}[vnumIndex]$
                \State $optL_v=Pool_{L_v}[L_vIndex]$
            \EndIf
        \EndIf
        \EndWhile 
    \EndWhile  
\EndProcedure

\end{algorithmic}
\end{algorithm}
\subsubsection*{Case Study}
With a specific pool of $(L_v,N_v)$,$L_v=[0.2 \  0.5 \  0.8 \  1.1 \  1.6 \  2.1 \  3 \  4 \  6]$ and $N_v=[1 \  2 \  3 \  4\  5 \  6]$, the analytical results are plotted as below in Figure ~\ref{fig:energydelaypair}. $\mu=0.8$ and $\lambda=0.3$.
\begin{figure}[!htbp]
\centering
    \includegraphics[width=0.98\linewidth]{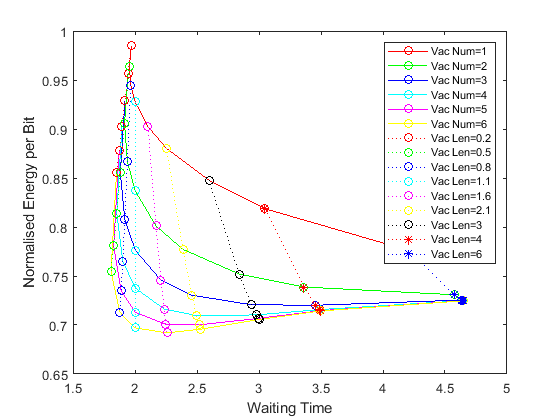}
\caption{Energy-Delay vs ($L_v$, $N_v$)}
\label{fig:energydelaypair}
\end{figure}
The plot is similar to Figure 3 in ~\cite{Kumar} when the vacation number is 1 and the vacation length is a constant. As in paper ~\cite{Kumar}, the traffic rate and service rate have not been mentioned for generating Figure 3, the comparison stops where the plots have similar curves but not exact values. With the increase of the vacation length, the normalised energy per bit decreases while the waiting time increases. And the same rule applies to the change of the vacation number.
Suppose the expected maximum delay is set to $2$, the derived optimal solution is $(0.8\ ,6)$.

To make the results convincing, simulation results are collected and brute force method is used to locate the ground truth and it can show in the figure below that the derived minimum is 2 steps away from the ground truth $(0.8 \ , 3)$ as shown in Figure ~\ref{fig:effectiveness}.The derived solution has a relative error rate of $[0.0299,0.022]$ from the ground truth value in this particular case study.
\begin{figure}[!htbp]
\centering
    \includegraphics[width=0.98\linewidth]{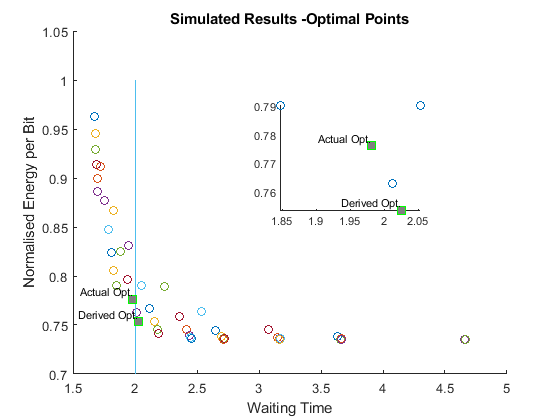}
\caption{Effectiveness of the Derived Solution}
\label{fig:effectiveness}
\end{figure}

\section{4-Stage Sleep Control Performance Evaluation}
Current state of the art of base station energy conservation design supports a vacation amount limit of 4 and each with different vacation length $[0.0000714\ 0.001\ 0.01\ 1]$ \cite{powermodel,5G}. The active power is 234.2 w, power when system is idle is 38.2 w, and power in the four individual sleep mode is $[25.5 w\ 2.9 w\ 2.0 w\ 1.8 w]$

The cases where input rate is solely being increased gradually and where the service rate is solely being increased gradually is investigated separately.
\subsection{Case 1}
From Figure\ref{fig:5GPandDv1} where $\mu=35025$, it can be concluded that with the increasing of traffic rate, the average waiting time is decreased. It is in alignment with Figure 6 in paper \cite{5G} in that with the increase of traffic intensity ($\rho \in [0\ 0.4]$ ), the normalised energy per bit rate increases and the normalised energy saving decreases.
\begin{figure}[!htbp]
\centering
    \includegraphics[width=0.98\linewidth]{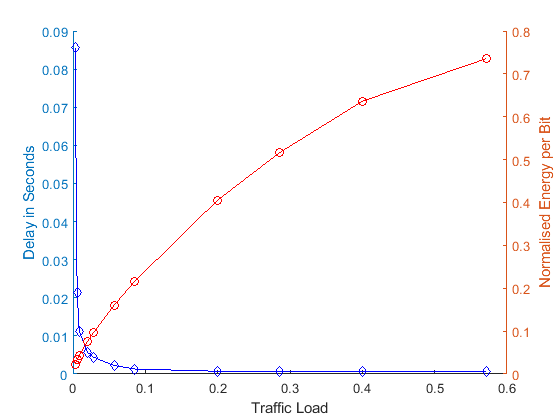}
\caption{System Performance against Input Rate}
\label{fig:5GPandDv1}
\end{figure}
\FloatBarrier
\subsection{Case 2}
Next, scenario where the service rate can be varied is being studied. Assume the system supports dynamic frequency scaling (the system voltage is fixed) where $p=\alpha V^2 f$ with $\alpha$ is a system specific constant. Let the data packet size to be uniform as $U$, which consumes $M$ clock cycles to process in total, then the time spent per packet is $MT_{clock}=\frac{M}{f_{clock}}$. It is inverse to the service rate so $\frac{f_{clock}}{M}=\mu$.That's $f_{clock}=M\mu$. Hence $p=\beta \mu$ with $\beta=\alpha V^2 M$.

For the ease of operation and observation, the configuration from Case 1 is borrowed where $\lambda=2000$,$\mu=35025$ and $p=234.2 w$. Using the formula above, $\beta=0.00668665$. Suppose $\lambda=2000$ is fixed, increase the traffic load $(0.005 \longrightarrow 0.4)$, results show that the energy consumption decreases as the traffic load decreases while the delay increases as the traffic load decreases. The former is reasonable because with less traffic load, the system enters sleep mode more often, the energy consumption is lessened henceforth. The latter is also reasonable as less traffic load means lower service rate and higher inter-departure time.

\begin{figure}[!htbp]
\centering
    \includegraphics[width=0.98\linewidth]{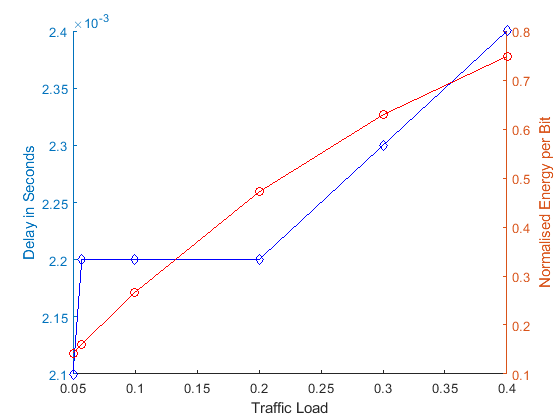}
\caption{System Performance against Service Rate with Input Rate fixed}
\label{fig:5GPandD}
\end{figure}
\section{CONCLUSIONS}
This work firstly discusses the advantage of the newly devised multi-stage sleep mode control for 5G system and then  presents a validated analytical model for it regarding energy efficiency and system delay. Moreover, the optimisation algorithm, as tested, guarantees to produce a solution that is deviated from the ground truth by minute discernible error.Lastly, the model is being applied to a realistic configuration to evaluate the system performance and the potential of selective service rate scheme has been tentatively investigated.
\section{Future Works}
Future works will investigate further into the delay modelling whether the discrepancy matters in practical engineering settings. The author will look into end-to-end delay bounds as specified in next generation mobile network. 
\section{Acknowledgement}
This work was funded by EPSRC grant EP/S00114X/1
\addtolength{\textheight}{-12cm}  



\end{document}